\documentclass[11pt]{article}
\usepackage{amsmath,amssymb,amsfonts,amsthm,epsfig}
\usepackage[usenames,dvipsnames]{xcolor}

\usepackage{bm,xspace}

\usepackage{cancel}
\usepackage{fullpage}
\usepackage{liyang}
\usepackage{framed}
\usepackage{verbatim}
\usepackage{enumitem}
\usepackage{array}
\usepackage{multirow}
\usepackage{afterpage}
\usepackage{mathrsfs}

\usepackage[normalem]{ulem}

\newcommand{\lnote}[1]{\footnote{{\bf \color{blue}Li-Yang}: {#1}}}
\newcommand{\rnote}[1]{\footnote{{\bf \color{red}Rocco}: {#1}}}

\usepackage{caption} 

\usepackage{todonotes}

\makeatletter
\newtheorem*{rep@theorem}{\rep@title}
\newcommand{\newreptheorem}[2]{
\newenvironment{rep#1}[1]{
 \def\rep@title{#2 \ref{##1}}
 \begin{rep@theorem}\itshape}
 {\end{rep@theorem}}}
\makeatother
\theoremstyle{plain}



\newcommand{\ignore}[1]{}

\def\colorful{1}

\ifnum\colorful=1

\newcommand{\gray}[1]{{\color{gray}{#1}}}

\fi
\ifnum\colorful=0

\newcommand{\gray}[1]{{{#1}}}

\fi

\usepackage{boxedminipage}

\newreptheorem{theorem}{Theorem}
\newtheorem*{theorem*}{Theorem}
\newreptheorem{lemma}{Lemma}
\newreptheorem{proposition}{Proposition}
\newtheorem*{noclaim*}{Claim}

\newcommand{\uhr}{\upharpoonright}

\newcommand{\acz}{\mathsf{AC^0}}

\newcommand{\ds}{\displaystyle}

\newcommand{\stars}{\mathrm{stars}}
\newcommand{\gentle}{\mathrm{gentle}}

\newcommand{\sand}{\mathrm{sand}}
\newcommand{\SL}{\mathrm{SL}}
\newcommand{\PRG}{\mathrm{PRG}}
\newcommand{\mcount}{\mathrm{count}} 
\newcommand{\simple}{\mathrm{simple}}

\usepackage{dsfont}
\renewcommand{\F}{\mathds{F}}
\renewcommand{\N}{\mathds{N}} 

\begin{document}

\title{
Deterministic search for CNF satisfying assignments\\ in almost polynomial time
}
\ignore{\rnote{
Alternate title possibilities:
``Fast deterministic search for CNF satisfying assignments'' 
``Derandomized CNF search in almost polynomial time''
``Derandomized search for CNF satisfying assignments in almost polynomial time''
``Derandomized CNF satisfiability in almost polynomial time''}
}
\author{ Rocco A.~Servedio\thanks{Supported by NSF grants CCF-1420349 and CCF-1563155. Email: {\tt rocco@cs.columbia.edu}}\\ 
Columbia University \and Li-Yang Tan\thanks{Supported by NSF grant CCF-1563122.  Part of this research was done during a visit to Columbia University. Email: {\tt liyang@cs.columbia.edu}} \\ Toyota Technological Institute}

\begin{titlepage}

\maketitle

\begin{abstract}
We consider the fundamental derandomization problem of deterministically finding a satisfying assignment to a CNF formula that has many satisfying assignments.  We give a deterministic algorithm which, given an $n$-variable $\poly(n)$-clause CNF formula $F$ that has at least $\eps 2^n$ satisfying assignments, runs in time
\[
n^{\tilde{O}(\log\log n)^2}
\]
for $\eps \ge 1/\polylog(n)$ and outputs a satisfying assignment of $F$.  Prior to our work the fastest known algorithm for this problem was simply to enumerate over all seeds of a pseudorandom generator for CNFs; using the best known PRGs for CNFs \cite{DETT10}, this takes time $n^{\tilde{\Omega}(\log n)}$ even for constant $\eps$.
Our approach is based on a new general framework relating deterministic search and deterministic approximate counting, which we believe may find further applications.

\end{abstract}

\thispagestyle{empty}

\end{titlepage}

\section{Introduction}

Understanding the role of randomness in efficient computation has been a major focus of complexity theory over the past several decades.  In particular, much effort has been dedicated to developing general techniques for \emph{unconditional derandomization}, i.e.~methods of constructing efficient deterministic algorithms (that do not rely on any unproven hardness assumptions) for computational problems that are known to have efficient randomized algorithms.\ignore{\gray{Toward this end, pseudorandom generators and deterministic algorithms for approximately counting satisfying assignments have been obtained for many different types of Boolean functions.}\lnote{Slight overlap between this and the next sentence? {\bf Rocco:}  I had been thinking that this gray sentence gets us from ``computational problems'' in the previous sentence (which could include things like finding matchings, primality testing, anything really) to satisfying-assignment-type problems --- if someone like Cliff as opposed to a circuit-head were reading this, it could be a sort of bridge.  But I think it's okay to omit it too.}}  Notable successes have been achieved in this line of work:  pseudorandom generators with highly non-trivial seed length, and much-faster-than-brute-force deterministic approximate counting algorithms, are now known for many function classes such as those defined by logarithmic space, small-depth circuits, sparse and low-degree $\F_2$ polynomials, various classes of branching programs, functions of a few halfspaces, low-degree polynomial threshold functions, and more (see e.g.~\cite{AjtaiWigderson:85,Nis91,LVW93,NW94,LV96,SZ99,Trevisan:04,Bra10,RS:10,GOWZ10,DGJ+10:bifh,DKN10,GKMSVV11,GMRTV12,IMZ12,Kane12,MZ13,TX13,DS14stoc,BRRY14,HS16} and many other works).

%\ignore{
%\rnote{Not sure we want to go this ``full derandomization'' route in the exposition, what do you think?  I do think a structure basically of this sort, where first para highlights some successes and 2nd para highlights limits of what's been achieved, is a reasonable structure though}}

While striking progress has thus been made, there remain fundamental gaps in our understanding of the overarching question in unconditional derandomization: can every randomized algorithm be made deterministic with only a polynomial slowdown?  In particular, while highly non-trivial results have been achieved for the classes mentioned above, a ``full derandomization''---i.e.~a deterministic algorithm running in \emph{polynomial time}, as opposed to, say, quasipolynomial time---remains elusive even for some of the simplest classes of functions.  (Even for the class of linear threshold functions, a full derandomization was only achieved in relatively recent work \cite{RS:10,GKMSVV11}.)

\vspace{-8pt}
\paragraph{The question we consider.}
Perhaps the most basic full derandomization problem that remains open is the \emph{CNF search problem}:  

\begin{quote}
\emph{Input:} An $n$-variable $M$-clause CNF formula $F$ that is promised to have many, say at least $\eps 2^n$, satisfying assignments.

\emph{Goal:} Output any satisfying assignment of $F$.

\end{quote}
Using randomness it is easy to find a satisfying assignment with high probability\ignore{ in time $\poly(n,M,1/\eps)$} simply by sampling $O(1/\eps)$ many assignments and evaluting $F$ on each one.  Is there a  polynomial-time \emph{deterministic} algorithm?  This problem was first considered by Ajtai and Wigderson in their pioneering work~\cite{AjtaiWigderson:85} on unconditional derandomization, in which they gave the first non-trivial (subexponential-time) deterministic algorithm for the problem.

\subsection{Prior results and related work} 

We briefly recall the prior state of the art for this and related problems.  
\vspace{-8pt}
\paragraph{Pseudorandom generators and hitting sets for CNFs.} Prior to our work the fastest known algorithm was simply to enumerate over all seeds of a pseudorandom generator $G$ that $\eps$-fools the class of $M$-clause $n$-variable CNF formulas; the definition of a pseudorandom generator immediately implies that some seed string $y$ will have $F(G(y))=1$.  Using the best known construction of $\eps$-PRGs for $M$-clause $n$-variable CNFs~\cite{DETT10}, this gives an algorithm running in time $\poly(n) \cdot (M/\eps)^{\tilde{O}(\log (M/ \eps))}$.  We observe that this PRG-based approach is oblivious to the input formula $F$, and can be used even if $F$ is only provided as a black-box oracle instead of an explicit CNF formula.  While this may be viewed as an advantage, it also suggests that non-oblivious approaches which exploit the structure of the input formula $F$ may be able to achieve faster runtimes.  We further observe that only an $\eps$-hitting set for CNFs rather than an $\eps$-PRG is required for this oblivious approach, but the best known explicit construction of hitting sets for general CNFs is simply the \cite{DETT10} PRG.  We recall that a seemingly-modest improvement of the~\cite{DETT10} PRG's seed length from $\tilde{O}(\log^2(M/\eps))$ to $O(\log^{1.99}(M/\eps))$, even for $\eps$-hitting sets, would improve state-of-the-art lower bounds against depth-three circuits, breaking a longstanding barrier in circuit complexity. (For the special case of \emph{read-once} CNF formulas, 
S{\'{\i}}ma and Z{\'{a}}k \cite{SimaZak10} have given an $\eps$-hitting set of $\poly(n)$ size for $\eps > 5/6$, and Gopalan et al.~\cite{GMRTV12} have given an $\eps$-PRG with seed length $\tilde{O}(\log(n/\eps))$.)

\vspace{-8pt} 
\paragraph{The work of Goldreich and Wigderson.} Recently, Goldreich and Widgderson~\cite{GW14} initiated the study of deterministic search in the regime where $\eps$ is extremely close to $1$, a relaxation of the standard regime where we typically think of $\eps = 1/2$ or $\eps = o(1)$. As one of their main results, they give a polynomial-time deterministic search algorithm for $\acz$ circuits when $\eps \ge 1-2^{n^{0.99}}/2^n$.  For the special case of $M$-clause $n$-variable CNF formulas (the subject of this work), they observe that if $\eps \geq 1-1/(4M)$ then any $\delta=1/(4M)$-biased sample space over $\{0,1\}^n$ must contain a satisfying assignment of $F$.  Since well-known deterministic algorithms \cite{NN93,AGHP92} can enumerate all $\poly(n/\delta)$ elements of such a sample space in $\poly(n/\delta)$ time, this gives a $\poly(n,M)$ time algorithm in this special case. (As they note in their paper, this observation is already implicit in the work of \cite{GMRTV12}.) 

\vspace{-8pt}
\paragraph{Deterministic approximate counting and answering Trevisan's question.} 
 While the PRG-based approach described above is the most efficient algorithm known for deterministic CNF search, a more efficient algorithm is known for deterministic \emph{approximate counting} of CNF satisfying assignments.  Building on early work of Luby and Veli{\v{c}}kovi{\'c}~\cite{LV96}, Gopalan, Meka, and Reingold~\cite{GMR13} gave a deterministic algorithm which, given as input an $M$-clause $n$-variable CNF $F$ and a parameter $\eps > 0$, runs in time $(Mn/\eps)^{\tilde{O}(\log \log n + \log \log M + \log(1/\eps))}$ and outputs an (additive) $\eps$-accurate estimate of the fraction of assignments that satisfy $F$.

 Trevisan \cite{Trevisan:slides10} has remarked that it is curious that this deterministic approximate counting algorithm---which in particular yields a certificate that $F$ has at least $\Omega(\eps 2^n)$ satisfying assignments---does not yield a comparably efficient algorithm to \emph{find} a satisfying assignment.  In \cite{Trevisan:slides10} he posed the problem of developing a deterministic search algorithm running in time comparable to that of deterministic approximate counting algorithms.  Our work gives a positive solution to this problem (though it should be noted that our search algorithm's exponent is roughly quadratic in the exponent of the \cite{GMR13} counting algorithm).

\subsection{Our main result and approach} 
\label{sec:our-approach} 
We give a deterministic CNF search algorithm that runs in almost polynomial time:

\begin{theorem} \label{thm:main}
There is a deterministic algorithm which, when given as input an $M$-clause CNF formula $F$ over $\{0,1\}^n$ that has $|F^{-1}(1)| \geq \eps 2^n$, runs in time
\[
\left({\frac {Mn} \eps}\right)^{\tilde{O}(\log \log (Mn) + \log (1/\eps))^2}
\]
and outputs a satisfying assignment of $F$.\ignore{\rnote{I believe this is what we prove, right?  If we prefer it aesthetically, we can instead state our main theorem as

``There is a deterministic algorithm which, when given as input an $M$-clause CNF formula $F$ over $\{0,1\}^n$ 
that has $|F^{-1}(1)| \geq \eps 2^n$, runs in time
\[
\left({\frac {Mn} \eps}\right)^{\tilde{O}(\log \log (Mn) + \log \log(1/\eps))^2}
\]
and outputs a satisfying assignment of $F$.''  (Note this never gives a weaker bound than Theorem \ref{thm:main}.)

Justification:  If $M \geq n^{\Omega(1)}$ then this follows immediately from Theorem \ref{thm:main}.  If $M$ is smaller, then we can pad $M$
with $n$ clauses $(x_1 \vee \overline{x}_1) \wedge \cdots \wedge (x_n \vee \overline{x}_n)$, thus achieving $M' = M+n = n^{\Omega(1)}$, and apply Theorem \ref{thm:main} to $M'.$  

(We could even state the running time as
\[
\left({\frac {M+n} \eps}\right)^{\tilde{O}(\log \log (M+n) + \log \log(1/\eps))^2}
\]
or as
\[
\left({\frac {\max\{M,n\}} \eps}\right)^{\tilde{O}(\log \log (\max\{M,n\}) + \log \log(1/\eps))^2}
\]

but those both look a bit silly.)

I'm happy with any way to present this that you prefer.
}}
\end{theorem}

For the case when $M=\poly(n)$ and $\eps \geq 1/\polylog(n)$, the running time of our algorithm is $n^{\tilde{O}(\log \log n)^2}$. As discussed above, the previous fastest algorithm takes time $n^{\tilde{\Omega}(\log n)}$ when $M = \poly(n)$, even for constant $\eps$. 

Our approach is based on a new general framework for obtaining deterministic search algorithms from deterministic approximate counting algorithms (given a few additional ingredients).  Roughly speaking, this approach is an extension of the generic naive reduction described in the next subsection; while the naive reduction constructs a satisfying assignment one coordinate at a time, our approach assigns a whole block of coordinates at each iteration as described in Section~\ref{sec:our-approach} below.  We are optimistic that this framework may find further applications for other deterministic search problems.

\subsubsection{Warm up: a simple and naive search algorithm based on approximate counting} To motivate our approach, we begin by considering a very simple and naive way of obtaining a deterministic search algorithm from a deterministic approximate counting algorithm. (We will specialize our discussion to the class of CNF formulas, but the generic reduction we describe here relates these two derandomization tasks for all function classes.)  Suppose we have a deterministic approximate counting algorithm $A_\mathrm{count}$ for the class of CNF formulas: given as input an $M$-clause CNF formula $F$ over $\zo^n$, this algorithm $A_\mathrm{count}$ runs in time $T(n,M,\delta)$ and outputs an (additive) $\delta$-accurate estimate of $\Pr[F(\bx)=1]$. Then this immediately yields, in a black-box manner, a deterministic search algorithm $A_\mathrm{search}$ with the following performance guarantee: given as input an $M$-clause CNF formula $F$ over $\zo^n$ that has $|F^{-1}(1)|\ge \eps 2^n$, the algorithm $A_\mathrm{search}$ runs in time 
\begin{equation} T(n,M,\eps/(4n)) \cdot 2n  \label{eq:naive-runtime} 
\end{equation} 
and outputs a satisfying assignment of $F$. The argument follows the standard $n$-stage decision-to-search reduction; in the $(i+1)$-st stage, after the first $i$ bits $(z_1,\dots,z_i) \in \zo^{{i}}$ have been obtained, the algorithm runs $A_\mathrm{count}$ with accuracy parameter $\delta := \eps/(4n)$ both on $F(z_1,\dots,z_i,0,x_{i +1},\dots,x_n)$ and on $F(z_1,\dots,z_i,1,x_{i+1},\dots,x_n)$, and takes as the next coordinate $z_{i+1}$ the bit corresponding to the higher output value from $A_\mathrm{count}.$  A straightforward induction shows that for all $i \in [n]$ we have
\[ 
\E\big[F(z_1,\dots,z_i,\bx_{i+1},\dots,\bx_n)\big] \ge
\E\big[F(\bx_1,\dots,\bx_n)\big] - 2 i \cdot \delta, 
\]
so the final string $(z_1,\dots,z_n)$ satisfies
$\E[F(z_1,\dots,z_n)]\ \geq \eps - {\frac \eps 2} > 0$ and hence $F(z_1,\dots,z_n)=1.$

However, instantiating this approach with the best known deterministic approximate counting algorithm due to Gopalan, Meka, and Reingold~\cite{GMR13}, which runs in time 
\[ T(m,M,\delta) = (Mn/\delta)^{\tilde{O}(\log \log n + \log \log M + \log(1/\delta))},\]
we see that the running time (\ref{eq:naive-runtime}) evaluates to  
\[ (Mn/\eps)^{\tilde{O}(\log(n/\eps) + \log \log M)}. \] 
This is $n^{\tilde{\Omega}(\log n)}$ when $M = \poly(n)$ (even for constant $\eps$), which is no improvement over the trivial PRG-based algorithm.  The crux of the problem with this naive approach is that we cannot afford to run the~\cite{GMR13} approximate counting algorithm to such high accuracy, $\delta = O(\eps/n)$. 

\subsubsection{Our approach:  a more efficient reduction}
\label{sec:our-approach} 

At the highest level, our search algorithm shares the same overall structure as the naive bit-by-bit approach sketched above. Our algorithm is recursive in nature and uncovers a satisfying assignment of $F$ in a stage-wise manner: in each stage we run a deterministic approximate counting algorithm on subfunctions of $F$, and we recurse on the one for which our estimate of its fraction of satisfying assignments is the largest.   However, instead of uncovering a single coordinate of a satisfying assignment per stage, our algorithm uncovers a \emph{$p$ fraction of the remaining coordinates} per stage where $p \gg 1/n$.  (In our analysis $p = \exp(-\Theta(\log\log(Mn/\eps))^2)$, though its precise value is unimportant for the rest of this high-level discussion.)  Roughly speaking, this allows us to circumvent the problem highlighted above since there will be at most $p^{-1}\ln n$ many stages in total (rather than $n$), and so in each stage we can run the~\cite{GMR13} approximate counting algorithm with a much larger error parameter $\delta = \Omega(\eps/(p^{-1}\ln n))$ instead of $\delta = O(\eps/n)$.  

\vspace{-8pt} 
\paragraph{Three main ingredients of our approach.}  We will describe our approach in general terms, since the overall framework is fairly versatile and could be instantiated in other contexts.  
\begin{itemize} 
\item Let $\calC$ be the function class of interest, the class for which we would like to design a deterministic algorithm for the ``$\calC$ search problem": given as input an $n$-variable function $F \in\calC$ that is promised to have at least $\eps2^n$ satisfying assignments, find a satisfying assignment. (Our analysis will assume that $\calC$ is closed under restrictions, which holds for natural function classes including the class of $M$-clause CNF formulas.)
\item Let $\calC_{\text{simple}} \sse \calC$ be a subclass of ``simple" functions within $\calC$.  \end{itemize} 

As alluded to above, the plan is to do search for $\calC$ recursively in stages, uncovering a satisfying assignment of $F \in \calC$ ``chunk-by-chunk". In each stage we employ three pseudorandom constructs, the first two of which are: 
\begin{enumerate} 
\item  A PRG for $\calC_{\text{simple}}$, and 
\item  A deterministic approximate counting algorithm $A_{\text{count}}$ for $\calC$.  
\end{enumerate} 
The win of our approach over the trivial PRG-based search algorithm will rely on both (1) the simplicity of the functions in $\calC_\simple$ enabling PRGs of significantly shorter seed length than those known for $\calC$, and in similar spirit, (2) the existence of an approximate counting algorithm for $\calC$ with runtime significantly better than that of the trivial PRG-based algorithm for $\calC$. 

The third and final ingredient is a ``pseudorandom $\calC$-to-$\calC_\simple$ simplification lemma": 
\begin{enumerate}
\item[3.]  Pseudorandom $\calC$-to-$\calC_\simple$ simplification lemma. 

Roughly speaking, such a simplification lemma says the following: there is a pseudorandom distribution $\calR$ over restrictions such that for all $F \in \calC$, with high probability over $\brho\leftarrow \calR$ the randomly restricted function $F \uhr \brho$ belongs to $\calC_\simple$.  In more detail, this pseudorandom distribution $\calR$ over the space of restrictions $\{0,1,\ast\}^n$ should have the following structure:  
\begin{enumerate}
\item The set of ``live" positions $\bL \sse [n]$ (i.e.~the set of $\ast$'s) can be sampled {efficiently} with seed length $r_\SL$.  We write $\bL \leftarrow \calR_\stars$ to denote a draw from this pseudorandom distribution over subsets of $[n]$. 
\item Non-live positions $[n]\setminus \bL$ are filled in independently and uniformly with $\{0,1\}$, and do not count against the seed length $r_\SL$. We write $\brho\leftarrow \zo^{[n]\setminus \bL}$ to denote a draw of such a restriction.
\end{enumerate} 

We will require each subset $L \in \supp(\calR_\stars)$ to have size at least $pn$ for some not-too-small $p \in (0,1)$ (equivalently, we will require $\calR$ to be supported on restrictions that leave at least a $p$ fraction of coordinates unfixed). As we will soon see, this is ``the same $p$'' as the $p$ in the high-level description of our approach in the first paragraph of this subsection; the size of $L$ corresponds exactly to the number of coordinates of a satisfying assignment that we uncover per stage. 

The guarantee that we will require of this pseudorandom $\calC$-to-$\calC_\simple$ simplification lemma is roughly as follows: for every $F \in \calC$, 
\begin{equation} \Ex_{\bL\leftarrow\calR_{\stars}}\bigg[ \Prx_{\brho\leftarrow\zo^{[n]\setminus \bL}}\big[\, (F \uhr \brho) \notin \calC_\simple\big] \bigg]  \le \delta_\SL, \label{eq:PSL}
\end{equation} 
where the failure probability $\delta_\SL$ is as small as possible.  In fact, our approach does not actually require that $F \uhr \brho$ belong to $\calC_\simple$; it suffices for $F \uhr \brho$ to be well-approximated by some $F' \in \calC_\simple$ for a suitable notion of approximation ($F \uhr \brho$ has a ``$\delta$-lower-approximator" in $\calC_\simple$).  The analysis of our CNF search algorithm will crucially exploit this relaxation of (\ref{eq:PSL}), but for clarity of exposition we will assume the stronger guarantee of (\ref{eq:PSL}) for the description of our general framework.  
\end{enumerate} 

For $\calC$ being the class of CNF formulas, we remark that ``pseudorandom $\calC$-to-$\calC_\simple$ simplification lemmas" have been the subject of much research~\cite{AjtaiWigderson:85,AAIPR01,IMP12,GMR13,TX13,GW14}. These simplification lemmas, more commonly referred to as \emph{pseudorandom switching lemmas} in this context, are achieved for various notions of ``simplicity", with $\calC_\simple$ being juntas~\cite{AjtaiWigderson:85,AAIPR01,IMP12,GW14}, decision trees~\cite{TX13}, or small-width CNF formulas~\cite{GMR13}.  We remark that for all these notions of ``simple" CNF formulas, there are indeed PRGs with significantly shorter seed length than the best known PRG for general CNF formulas~\cite{DETT10}. (In our analysis $\calC_\simple$ will be the class of $(\log((\log Mn)/\eps))$-width CNF formulas, as this leads to the best overall parameters in our final result.)

Going back to the general framework, we now explain how these three pseudorandom constructs---(1) PRG for $\calC_\simple$, (2) deterministic approximate counting algorithm $A_\mcount$ for $\calC$, and (3) pseudorandom $\calC$-to-$\calC_\simple$ simplification lemma---fit together to give a deterministic search algorithm for $\calC$. 

\vspace{-8pt}
\paragraph{A simple but crucial fact from~\cite{AjtaiWigderson:85}.}  At the heart of our analysis is an elementary fact about pseudorandom simplification lemmas. This fact was first stated and utilized in the influential work of Ajtai and Wigderson~\cite{AjtaiWigderson:85} giving the first non-trivial PRG for $\acz$ circuits; variants of it also play a role in the more recent PRG constructions of~\cite{GMRTV12,IMZ12,RSV13,TX13}.  

 Suppose that we have a pseudorandom $\calC$-to-$\calC_\simple$ simplification lemma satisfying (\ref{eq:PSL}). Fix an $L \in \supp(\calR_\stars)$ such that the inner probability of (\ref{eq:PSL}) is at most $\delta_\SL$.  Let $\calD$ be a distribution that $\delta_\PRG$-fools $\calC_\simple$, and suppose $\calD$ can be sampled with $r_\PRG$ many random bits. The simple but crucial fact from~\cite{AjtaiWigderson:85} is the following: the distribution over $\zo^n$ where 
\begin{enumerate}
\item The coordinates in $[n] \setminus L$ are filled in with uniform random bits;
\item The coordinates in $L$ are filled in according to the pseudorandom distribution $\calD$,
\end{enumerate}
$(\delta_\SL + \delta_\PRG)$-fools $\calC$. That is, for all $F \in \calC$, 
\[ \mathop{\Ex_{\bx\leftarrow\calU}}_{\by\leftarrow\calD}\big[ F(\bx_{[n]\setminus L}, \by_L)\big] = \Ex_{\bx\leftarrow\calU}\big[F(\bx)\big] \pm (\delta_\SL + \delta_\PRG). \] 
Given this observation of~\cite{AjtaiWigderson:85}, it follows that there must exist at least one $y \in \supp(\calD)$ such that 
\[ \Ex_{\bx\leftarrow\calU}\big[ F(\bx_{[n]\setminus L}, y_L)\big] \ge \Ex_{\bx\leftarrow\calU}\big[F(\bx)\big] - (\delta_\SL + \delta_\PRG). \] 
Equivalently, the restriction $\pi^*$ that fixes the coordinates in $L$ according to $y$ preserves (from below) $F$'s fraction of satisfying assignments up to an error of $(\delta_\SL + \delta_\PRG)$, by which we mean: 
\begin{equation} \Ex_{\bx\leftarrow\calU}\big[ (F \uhr \pi^*)(\bx)\big] \ge \Ex_{\bx\leftarrow\calU}\big[F(\bx)\big] - (\delta_\SL + \delta_\PRG). \label{eq:good-pi} 
\end{equation} 
Note that the number of coordinates that $\pi^*$ fixes is precisely the size of $L$, which explains why, as alluded to above, we  require the pseudorandom simplification lemma to be such that every $L \in \supp(\calR_\stars)$ has size at least $pn$ for some not-too-small $p\in (0,1)$. 

\vspace{-8pt}
\paragraph{Our search algorithm and its analysis.}  Our goal in a single stage of the recursive algorithm is to find a restriction that (approximately) satisfies (\ref{eq:good-pi}): such a restriction reduces our search space $\zo^n$ by $|{\pi^*}^{-1}(\{0,1\})| = |L| \ge pn$ many dimensions, while ensuring that the restricted function $F \uhr \pi^*$ still has ``many" satisfying assignments.

To accomplish this, our search algorithm cycles through all $2^{r_\SL + r_\PRG}$ candidates $\pi$---that is, all possible restrictions fixing $L$ according to $y$ where $L \in \supp(\calR_\stars)$ and $y \in \supp(\calD)$---and  for each candidate $\pi$, it runs the deterministic approximate counting algorithm $A_\mcount$ to estimate $\E[(F\uhr \pi)(\bx)]$ to accuracy $\delta_\mcount$.  It is straightforward to see that the restriction $\tilde{\pi}$ for which $A_\mcount$'s estimate is the largest will satisfy 
\[ \Ex_{\bx\leftarrow\calU}\big[ (F \uhr \tilde{\pi})(\bx)\big] \ge \Ex_{\bx\leftarrow\calU}\big[F(\bx)\big] - (\delta_\SL + \delta_\PRG) - 2\delta_\mcount. \]  
Up to an additive factor of $2\delta_\mcount$, this restriction $\tilde{\pi}$ is ``as good as" the restriction $\pi^{*}$ from (\ref{eq:good-pi}). 
Our algorithm recurses on $F \uhr \tilde{\pi}$, a function over $\zo^{\tilde{\pi}^{-1}(\ast)}$ where $|\tilde{\pi}^{-1}(\ast)| \le (1-p)n$. 
The runtime of this single stage of our recursive algorithm is at most 
\[ 2^{r_\SL + r_\PRG} \cdot T(n,\delta_\mcount), \]
where $T(n,\delta)$ denotes the running time of the deterministic approximate counting algorithm $A_\mcount$, when given as input an $n$-variable function $F \in \calC$ and accuracy parameter $\delta$. 

By fixing at least a $p$ fraction of the remaining coordinates in each stage, we ensure that there are at most $p^{-1}\ln n$ many stages in total, after which all $n$ coordinates will have been fixed to a certain assignment $x \in \zo^n$ and the algorithm terminates with $x$ as its output.  Hence, by choosing parameters so that 
\[ \delta_\SL + \delta_\PRG + 2\delta_\mcount \le \frac1{2} \cdot \frac{\eps}{p^{-1}\ln n},  \] 
we ensure that the algorithm always recurses on a subfunction that is satisfied by at least an $(\eps/2)$-fraction of its assignments. In particular, this guarantees that the $n$-bit assignment $x \in \zo^n$ which the algorithm outputs is indeed a satisfying assignment of $F$.  The overall runtime of the entire algorithm is 
\[ 2^{r_\SL + r_\PRG} \cdot T(n,\delta_\mcount) \cdot p^{-1}\ln n. \]

\subsection{Organization of this paper} 

In the rest of this paper we instantiate the general framework described above with $\calC$ being the class of $M$-clause $n$-variable CNF formulas, thus establishing our main result (Theorem~\ref{thm:main}).

In Section~\ref{sec:prelim} we recall the relevant definitions and state a few simplifying assumptions.  In Section~\ref{sec:PSL} we state the pseudorandom $\calC$-to-$\calC_\simple$ simplification lemma that we will use in our context (the pseudorandom switching lemma of~\cite{GMR13}, with $\calC_\simple$ being the class of small-width CNF formulas) and establish some of its basic properties.  In Sections~\ref{sec:AW} and~\ref{sec:existence} we use an extension of the~\cite{AjtaiWigderson:85} fact, together with this pseudorandom switching lemma and a PRG for $\calC_\simple$, to construct a small set of restrictions that is guaranteed to contain a ``good" restriction $\pi^{*}$, one that fixes a significant fraction of coordinates while preserving the bias of a CNF formula from below.   In Section~\ref{sec:find} we show how to use a deterministic approximate counting algorithm to search through this set and find a restriction $\tilde{\pi}$ that is ``almost as good as" $\pi^*$, thus completing the description of one stage of our recursive search algorithm.  Finally, in Section~\ref{sec:done} we put the pieces together and give our overall recursive search algorithm.

\section{Background and setup} \label{sec:prelim}

For $r < n$, we say that a distribution $\calD$ over $\zo^n$ can be \emph{sampled efficiently with $r$ random bits} if (i) $\calD$ is the uniform distribution over a multiset of size exactly $2^r$ of strings from $\zo^n$, and (ii) there is a deterministic algorithm $\mathrm{Gen}_{\calD}$ which, given as input a uniform random $r$-bit string $\bx \leftarrow \zo^r$, runs in time $\poly(n)$ and outputs a string drawn from $\calD$.%\rnote{this OK?}   

For $\delta>0$ and a class $\calC$ of functions from $\zo^n$ to $\zo$, we say that a distribution $\calD$ over $\zo^n$ \emph{$\delta$-fools $\calC$ with seed length $r$} if (a) $\calD$ can be sampled efficiently with $r$ random bits via algorithm $\mathrm{Gen}_\calD$, and (b) for every function $f \in \calC$, we have
\[
\bigg|\Ex_{\bs \leftarrow \{0,1\}^r}[f(\mathrm{Gen}_{\calD}(\bs))] - 
\Ex_{\bx \leftarrow \{0,1\}^n}[f(\bx)]\bigg| \leq \delta.
\]
Equivalently, we say that $\mathrm{Gen}_\calD$ is a \emph{$\delta$-PRG for $\calC$ with seed length $r$.}

Given a function $f: \zo^n \to \zo$ and a class of functions $\calC$ from $\zo^n$ to $\zo$, we say that \emph{$f$ is $\delta$-sandwiched by $\calC$} if there exist functions $g_\ell, g_u \in \calC$ such that (i) $g_\ell(x) \leq f(x) \leq g_u(x)$ for all $x \in \zo^n$, and (ii) $\E_{\bx \leftarrow \zo^n}[g_u(\bx) - g_\ell(\bx)] \leq \delta.$  The function $g_\ell$ ($g_u$, respectively) is said to be a \emph{lower $\delta$-approximator} (\emph{upper $\delta$-approximator}, respectively) for $f$.

\paragraph{Some simplifying assumptions.} We first observe that we may assume without loss of generality that our algorithm is given the value of $\eps$.  This is because the algorithm can try values $\eps={\frac 1 2}, {\frac 1 4}, {\frac 1 8}, \cdots$, halting when it finds a satisfying assignment, without changing the claimed asymptotic running time. We next observe that we may assume without loss of generality that the input CNF formula has $M \geq n$ many clauses.
This is because  if $M<n$  then we can pad $F$ with $n-M$ clauses $(x_1 \vee \overline{x}_1), \cdots, (x_{n-M} \vee \overline{x}_{n-M})$ to obtain an equivalent formula $F'$ with $n$ clauses and run the algorithm on $F'$.

The following simple observation allows us to assume without loss of generality that the $M$-clause input CNF formula has width bounded by $O(\log(M/\eps))$: 

\begin{observation}[Trimming $F$]
\label{obs:trim} 
\emph{
Let $F$ be an $M$-clause CNF over $\{0,1\}^n$, and let $F'$ be the CNF obtained from $F$ by trimming each clause of width $w' > w := \log(2M/\eps)$ to width exactly $w$ (by removing an arbitrary $w'-w$ literals from the clause). Then
\begin{enumerate}
\item $F'^{-1}(1) \sse F^{-1}(1)$. 
\item $\Ex[F'(\bx)] \ge \Ex[F(\bx)] - \eps/2$. 
\end{enumerate} 
We observe that $F'$ can be constructed deterministically from $F$ in time $\poly(M,n)$. }
\end{observation} 

\begin{proof}  The observation about efficiently constructing $F'$ from $F$ is immediate, as is part (1) since if an assignment satisfies a given clause of $F'$ then clearly it satisfies the corresponding clause of $F$.  Part (2) holds because each time a clause is replaced by its trimmed version, the total number of satisfying assignments is reduced by at most $2^{-w}\cdot 2^n = {\frac \eps {2M}} \cdot 2^n$.
\end{proof} 

Our algorithm will begin by trimming all wide clauses of $F$ (of width greater than $\log(2M/\eps)$) to have width exactly $\log(2M/\eps)$. By Observation~\ref{obs:trim}, if $F$ is $\eps$-satisfiable then the resulting $F'$ remains $(\eps/2)$-satisfiable, and furthermore any satisfying assignment of $F'$ is a satisfying assignment of the original CNF $F$.  

Combining all of the simple observations in this section, in order to prove Theorem~\ref{thm:main} it suffices to prove the following:

\begin{theorem} \label{thm:main2}
There is a deterministic algorithm with the following properties:  It is given as input a value $\eps>0$ and a CNF formula $F$ over $\{0,1\}^n$ with $M\geq n$ clauses, each of width at most $O(\log(M/\eps))$, such that $|F^{-1}(1)| \geq \eps 2^n$.  The algorithm runs in time
\[
\left({\frac {M} \eps}\right)^{\tilde{O}(\log \log M +  \log(1/\eps))^2}
\]
and outputs a satisfying assignment of $F$.
\end{theorem}

In the rest of the paper we prove Theorem~\ref{thm:main2} (so the number of clauses $M$ is assumed to be at least $n$ throughout the rest of the paper).
%\rnote{Made a pass to try and eliminate all ``$\log(Mn)$'' type expressions from the rest of the paper as they are now unnecessary, hopefully this was done correctly.}

\section{The \cite{GMR13} pseudorandom switching lemma}
\label{sec:PSL}

As outlined in Section~\ref{sec:our-approach}, one of the main ingredients of our deterministic search framework is a ``pseudorandom $\calC$-to-$\calC_\simple$ simplification lemma". For $\calC$ being the class of CNF formulas, these are more commonly known as \emph{pseudorandom switching lemmas}---randomness efficient versions of the seminal switching lemmas~\cite{FSS:84,ajtai1983,yao1985,Hastad86} from circuit complexity---and they have been the subject of much research~\cite{AjtaiWigderson:85,AAIPR01,IMP12,GMR13,TX13,GW14}.

We will use a recent pseudorandom switching lemma of Gopalan et al.~\cite{GMR13} as it leads to the best overall running time.  In this pseudorandom switching lemma  $\calC_\simple$ is the class of ``narrow" (width-$w'$) CNFs. As alluded to in Section~\ref{sec:our-approach}, this is not quite a pseudorandom $\calC$-to-$\calC_\simple$ simplification lemma in the sense of~(\ref{eq:PSL}): rather than showing that $F \uhr \brho$ belongs to $\calC_\simple$ with high probability, the~\cite{GMR13} pseudorandom switching lemma only guarantees that $F \uhr \brho$ is \emph{sandwiched} by $F_{\text{upper}}, F_{\text{lower}} \in \calC_\simple$ with high probability. But as we show in the next section, the analysis we sketched in Section~\ref{sec:our-approach} extends to accommodate this; in fact, for our purposes it suffices for $F \uhr \brho$ just to have a lower sandwiching approximator in $\calC_\simple$. 

We recall a standard definition from pseudorandomness:

%\rnote{Add some more exposition to this section; it will depend on the exposition we do in the intro}

\begin{definition}[$p$-regular distributions]
A distribution $\calR_\stars$ over subsets of $[n]$ is said to be \emph{$p$-regular} if for each $i \in [n]$ we have $\Prx_{\boldsymbol{L} \leftarrow \calR_\stars}[i\in \boldsymbol{L}] = p.$
\end{definition} 

Our deterministic search framework requires that the pseudorandom $\calC$-to-$\calC_\simple$ simplification lemma holds with respect to a distribution over restrictions with the following structure: first a draw from a pseudorandom distribution $\calR_\stars$ selects a subset $\bL \sse [n]$ of coordinates which will ``receive $\ast$'s" (the $\bL$ive coordinates), and then the non-$\ast$ coordinates $[n] \setminus \bL$ are filled in uniformly at random with bits. The~\cite{GMR13} pseudorandom switching lemma satisfies this prescribed structure:

%As the name suggests, we think of a draw of $\bL \leftarrow \calR_\stars$ as selecting a subset of coordinates which will ``receive $\ast$'s'' under a pseudorandom distribution of restrictions.\ignore{\rnote{Make sure to have said already by this point that we view a restriction $\rho$ as an element of $\{0,1,\ast\}^n$.}}  The idea of the following pseudorandom switching lemma from \cite{GMR13} is that if the non-$\ast$ coordinates \violet{$[n]\setminus \bL$} are filled in uniformly at random \violet{with bits}, then \gray{under a suitable pseudorandom restriction,} any width-$w$ CNF $F$ is very likely to ``simplify,'' in the sense that it can be sandwiched by ``narrow'' (width-$w'$) CNFs.

\begin{theorem}[{Theorem 5.3 of} \cite{GMR13}, pseudorandom switching lemma]
\label{thm:GMR-SL}
There is a universal constant $C > 0$ such that for all $w, w', \delta_\sand, \eta > 0$ and all $p$ satisfying 
\begin{equation} p \le \frac{\eta}{(w \log(1/\delta_{\sand}))^{C\log w}}, \label{eq:p}
\end{equation}
there is a $p$-regular distribution $\calR_\stars$ over subsets of $[n]$ that can be sampled efficiently using $r_\SL$ random bits where 
\begin{equation} r_\SL = O((\log w)(\log n + w' \log((\log w)/\eta)) + w \log(w \log(1/\delta_\sand))) \label{eq:r-SL}
\end{equation} 
\ignore{
\rnote{Was
$$
 r_\SL = O((\log w)(\log n + w' \log(1/\eta)) + w \log(w \log(1/\delta_\sand)))
$$
}
}and the following holds: for any width-$w$ CNF $F$, 
\[ \mathop{\Prx_{\bL\leftarrow\calR_\stars}}_{\brho\leftarrow \zo^{[n]\setminus \boldsymbol{L}}}[\, F\uhr \brho \text{~is not $\delta_{\sand}$-sandwiched by width-$w'$ CNFs}\,] \le \delta_\sand + \eta^{w'/4}.\] 
\end{theorem} 

We require pseudorandom restrictions that do not put down too few $\ast$'s.  This motivates the following corollary:

\begin{corollary}[Condition on having sufficiently many stars]
\label{cor:GMR-SL}
For the distribution $\calR_{\stars}$ defined in Theorem~\ref{thm:GMR-SL}, let $\calR_\stars'$ denote the distribution of $\bL \leftarrow\supp(\calR_\stars)$ conditioned on $\bL$ satisfying $|\bL|\ge pn/2$. Then for any width-$w$ CNF $F$,
 \[ \mathop{\Prx_{\bL\leftarrow\calR'_\stars}}_{\brho\leftarrow \zo^{[n]\setminus \boldsymbol{L}}}[\, F\uhr \brho \text{~is not $\delta_\sand$-sandwiched by width-$w'$ CNFs}\,] \le \frac{2(\delta_\sand + \eta^{w'/4})}{p}.\]
\end{corollary}

\begin{proof} 
Since $\calR_\stars$ is $p$-regular we have that $\Ex_{\boldsymbol{L}\leftarrow \calR_\stars}[|\boldsymbol{L}|] = pn$, and so
\[ \Prx_{\boldsymbol{L}\leftarrow\calR_\stars}[\boldsymbol{L} \in \supp(\calR'_\stars)]  = \Prx_{\boldsymbol{L}\leftarrow\calR_\stars}\Big[|\boldsymbol{L}| \ge \frac{pn}{2}\Big] \ge \frac{p}{2}.\] 
Hence 
\begin{align*}
 \mathop{\Prx_{\bL\leftarrow\calR_\stars'}}_{\brho\leftarrow \zo^{[n]\setminus \boldsymbol{L}}}[\, F\uhr \brho \text{~is not $\delta_{\text{sand}}$} \cdots \,] 
 &=  \mathop{\Prx_{\bL\leftarrow\calR_\stars}}_{\brho\leftarrow \zo^{[n]\setminus \boldsymbol{L}}}[\, F\uhr \brho \text{~is not $\delta_{\text{sand}}$} \cdots  \mid \boldsymbol{L}\in \supp(\calR'_\stars)\,] \\
 &\le  \mathop{\Prx_{\bL\leftarrow\calR_\stars}}_{\brho\leftarrow \zo^{[n]\setminus \boldsymbol{L}}}[\, F\uhr \brho \text{~is not $\delta_{\text{sand}}$} \cdots \,]\cdot \frac1{\Pr[\boldsymbol{L}\in \supp(\calR'_\stars)]}  \\
 &\le  \frac{2(\delta_{\text{sand}} + \eta^{w'/4})}{p}.
\end{align*}
\vskip -18pt
\end{proof} 

\section{Bias preservation via pseudorandom switching lemmas}
\label{sec:AW} 

An important ingredient in our approach is a simple but ingenious observation due to Ajtai and Wigderson \cite{AjtaiWigderson:85} which we state and prove as Lemma \ref{lem:TX-switcheroo} below.  Informally, it says the following:
Let $F : \zo^n \to \zo$ be a Boolean function and suppose there is a partition of $[n]$ into $L$ and $[n] \setminus L$ with the following property: with high probability over a uniform random restriction $\brho$ fixing the coordinates in $[n] \setminus L$ and leaving
the coordinates in $L$ free, the function $F \uhr \brho$ falls into a class $\calC_\simple$ that is fooled by a distribution $\calD$ over $\{0,1\}^n$. Then the pseudorandom distribution over restrictions that fixes the coordinates in $L$ according to $\calD$ and leaves coordinates in $[n] \setminus L$ free approximately preserves the bias of $F$.

In fact, in our analysis we will only require that $F \uhr \brho$ has a lower approximator in $\calC_\simple$. This is because for our purposes (deterministic search) it suffices to approximately preserve the bias of $F$ only in one direction: we have to ensure that the bias of $F$ does not decrease by too much (so that we do not lose too many or all of the satisfying assignments), but we are fine if the bias increases.

\begin{lemma}[Implicit in \cite{AjtaiWigderson:85}]
\label{lem:TX-switcheroo}
Let $F :\zo^n \to \zo$ and $L \sse [n]$.
Fix a class $\calC_\simple$ of functions over $\{0,1\}^{n}$ and let $\calD$ be a distribution over $\{0,1\}^{n}$ that $\delta_\PRG$-fools $\calC_\simple$. Suppose that 
\begin{equation} \Prx_{\brho\leftarrow \zo^{[n]\setminus L}}[\, F\uhr \brho\text{~does not have a lower $\delta_\sand$-approximator in $\calC_\simple$}\,] \le \delta_\SL. \label{eq:SL-assumption}
\end{equation}
Then 
\[ \mathop{\Ex_{\bx\leftarrow\calU}}_{\by\leftarrow\calD}[ F(\bx_{[n] \setminus L},\by_L) ] \ge \Ex_{\bx\leftarrow\calU}[F(\bx)] - (\delta_\PRG + \delta_\sand + \delta_\SL).  \] 
\end{lemma} 

\begin{proof}
If $F \uhr \rho$ has a lower $\delta_{\text{sand}}$-approximator $F' \in \calC_\simple$ then 
\begin{align*} 
\Ex_{\bx\leftarrow\calU}[ (F\uhr \rho)(\bx)] &\le \Ex_{\bx\leftarrow\calU}[F'(\bx)] + \delta_{\text{sand}} \tag*{($F'$ is a $\delta_\sand$-approximator for $F \uhr \rho$)}\\
&\le \bigg(\Ex_{\by\leftarrow\calD}[F'(\by)] + \delta_{\text{PRG}}\bigg) + \delta_{\text{sand}} \tag*{($\calD$ $\delta_{\text{PRG}}$-fools $F'$)} \\
&\le \bigg(\Ex_{\by\leftarrow\calD}[(F\uhr \rho)(\by)] + \delta_{\text{PRG}}\bigg) + \delta_{\text{sand}}, \tag*{(${F'} \le (F\uhr \rho)$ \text{~pointwise})} 
\end{align*}
and so 
\begin{align*}
\Ex_{\bx\leftarrow\calU}[F(\bx)] &= \Ex_{\brho\leftarrow\zo^{[n]\setminus L}}\Big[\Ex_{\bx\leftarrow\calU}[(F\uhr\brho)(\bx)]\Big] \\
&\le \bigg(\Ex_{\brho\leftarrow\zo^{[n]\setminus L}}\Big[\Ex_{\by\leftarrow\calD}[(F\uhr \brho)(\by)] \Big] + \delta_{\text{PRG}} + \delta_{\text{sand}}\bigg)  + \delta_{\text{SL}}  \tag*{((\ref{eq:SL-assumption}) and above)}\\
&=  \mathop{\Ex_{\bx\leftarrow\calU}}_{\by\leftarrow\calD}[ F(\bx_{[n] \setminus L},\by_L) ] + (\delta_{\text{PRG}} + \delta_{\text{sand}} + \delta_{\text{SL}}).
\end{align*}
This completes the proof. 
\end{proof} 

We will apply Lemma~\ref{lem:TX-switcheroo} with $\calC_\simple$ being the class of width-$w'$ CNFs (we will keep $w'$ a free parameter for now, but looking ahead we will ultimately set $w' = \Theta(\log w + \log((\log M)/\eps))$),
and $\calD$ being the distribution given by~\cite{GMR13}'s pseudorandom generator: 

\begin{theorem}[Theorem 3.1 of \cite{GMR13}, PRG for width-$w'$ CNFs]
\label{thm:GMR-PRG}
The class of width-$w'$ CNFs over $\{0,1\}^n$ can be $\delta_\PRG$-fooled by a distribution $\calD_\PRG$ which can be sampled with
\begin{equation} r_\PRG = O((w')^2 (\log(w' \log(1/\delta_\PRG)))^2 + w'\log(w')\log(1/\delta_\PRG) + \log\log n).\label{eq:r-PRG}
\end{equation} 
random bits. 
\end{theorem}
(We remark that the~\cite{DETT10} PRG for width-$w'$ $M$-clause CNFs can be used in place of Theorem~\ref{thm:GMR-PRG} in our analysis, and will result the same overall running time.)
\section{Existence of a bias-preserving restriction $\pi^*$} 
\label{sec:existence} 

We are ready to combine the results from the previous sections to prove the key structural fact underlying our search algorithm. Roughly speaking, the next lemma says that there is a small set of restrictions, all of which fix a significant fraction of coordinates, such that for every width-$w$ CNF $F$ there is at least one restriction in this set that approximately preserves the bias of $F$ from below. 

\begin{lemma}[Existence of a bias-preserving restriction]
\label{lem:key-lemma} 
For all $w, w', \delta_\sand, \delta_\PRG, \eta > 0$ and all $p$ satisfying assumption (\ref{eq:p}) of Theorem~\ref{thm:GMR-SL}, there is a distribution $\calR_\gentle$ over restrictions in $\{0,1,\ast\}^n$ such that the following hold: 
\begin{enumerate}
\item $\calR_\gentle$ is uniform over a multiset of at most $2^{r_\SL + r_\PRG}$ many outcomes, where 
\begin{align*}
r_\SL &= O((\log w)(\log n + w' \log((\log w)/\eta)) + w \log(w \log(1/\delta_\sand)))\\
r_\PRG &= O((w')^2 (\log(w' \log(1/\delta_\PRG)))^2 + w'\log(w')\log(1/\delta_\PRG) + \log\log n).
\end{align*} 
\item $|\pi^{-1}(\{0,1\})| \ge pn/2$ for all $\pi \in \supp(\calR_\gentle)$. 
\item For any width-$w$ CNF $F$ over $\{0,1\}^n$, there is at least one $\pi^* \in \supp(\calR_\gentle)$ such that 
\begin{equation} \Ex_{\bx\leftarrow\calU}[(F\uhr \pi^*)(\bx)] \ge \Ex_{\bx\leftarrow\calU}[F(\bx)] - (\delta_\PRG + \delta_\sand + \delta_\SL),\label{eq:bias-preservation}
\end{equation}
where 
\[ \delta_\SL =  \frac{2(\delta_\sand + \eta^{w'/4})}{p}.\] 
\end{enumerate} 
\end{lemma}

\begin{proof} 
The distribution $\calR_\gentle$ is defined as follows: to make a draw $\bpi \leftarrow\calR_\gentle$, 
\begin{enumerate}
\item Draw $\bL \leftarrow\calR'_\stars$, the distribution over subsets $L \sse [n]$ defined in Corollary~\ref{cor:GMR-SL}. 
\item Draw $\by\leftarrow\calD_\PRG$, the distribution over $\{0,1\}^n$ from Theorem~\ref{thm:GMR-PRG} that $\delta_\PRG$-fools $w'$-CNFs.
\item Output the restriction $\bpi\in \{0,1,\ast\}^n$ where
\[ 
\bpi_i = \begin{cases}
\by_i & \text{if $i \in \bL$}\\
\ast & \text{otherwise.}
\end{cases} 
\] 
\end{enumerate}
By Corollary~\ref{cor:GMR-SL} and Theorem~\ref{thm:GMR-PRG}, we have that $\calR'_\stars$ is uniform over a multiset of at most $2^{r_\SL}$ outcomes
and $\calD_\PRG$ is uniform over a multiset of $2^{r_\PRG}$ many outcomes, and hence $\calR_\gentle$ is uniform over a multiset of at most $2^{r_\SL + r_\PRG}$ many outcomes. By its definition, the distribution $\calR'_\stars$ satisfies $|L| \ge pn/2$ for all $L \in \supp(\calR'_\stars)$, and hence $|\pi^{-1}(\{0,1\})| \ge pn/2$ for all $\pi \in \supp(\calR_\gentle)$. 

It remains to justify the third claim above. For any width-$w$ CNF $F$, by Corollary~\ref{cor:GMR-SL} there must be at least one $L \in \supp(\calR'_\stars)$ that satisfies the assumption (\ref{eq:SL-assumption}) of Lemma~\ref{lem:TX-switcheroo} with $\calC_\simple$ being the class of width-$w'$ CNFs and $\delta_\SL = 2(\delta_\sand + \eta^{w'/4})/p$. For such an $L$, it follows from Lemma~\ref{lem:TX-switcheroo} and the fact that $\calD_\PRG$ $\delta_\PRG$-fools $\calC_\simple$ that  
\[ \mathop{\Ex_{\bx\leftarrow\calU}}_{\by\leftarrow\calD_\PRG}[ F(\bx_{[n] \setminus L},\by_L) ] \ge \Ex_{\bx\leftarrow\calU}[F(\bx)] - (\delta_\PRG + \delta_\sand + \delta_\SL),  \]
and hence  
\[ \Ex_{\bx\leftarrow\calU}[ F(\bx_{[n] \setminus L},y_L) ] \ge \Ex_{\bx\leftarrow\calU}[F(\bx)] - (\delta_\PRG + \delta_\sand + \delta_\SL)  \]
for at least one $y \in \supp(\calD_\PRG)$. This pair $(y,L)$ therefore defines a restriction $\pi^* \in \supp(\calR_\gentle)$---the restriction that fixes the coordinates in $L$ according to $y$---that satisfies (\ref{eq:bias-preservation}), and the proof is complete.
\end{proof} 

\section{Finding $\pi^*$, or a restriction $\tilde{\pi}$ that is almost as good} 
\label{sec:find} 

 To find a restriction that (approximately) satisfies (\ref{eq:bias-preservation}) we will approximate the bias of $F \uhr \pi$ for all candidates $\pi \in \supp(\calR_\gentle)$ using a deterministic approximate counting algorithm for CNF formulas:
\begin{theorem}[Theorem~4.6 of~\cite{GMR13} (second equation before end of proof), approximate counting algorithm]
\label{thm:GMR-count}
There is a deterministic algorithm that runs in time
\[ T_\mcount = M n^{O(\log(w/\delta_\mcount))}(\log n)^{O(w)} 2^{O(w\log(w/\delta_\mcount)(\log\log(w/\delta_\mcount))^2)}\] 
and
$\delta_\mcount$-approximates the bias of any $M$-clause width-$w$ CNF $F$ over $\{0,1\}^n$, i.e. it outputs a value 
$v \in [0,1]$ such that $|v - \E_{\bx \leftarrow \zo^n}[F(\bx)]| \leq \delta_\mcount.$
\end{theorem}

Combining Lemma~\ref{lem:key-lemma} and Theorem~\ref{thm:GMR-count}, we get: 

\begin{corollary}[One stage of our recursive algorithm]
\label{cor:single-step} 
There is a deterministic algorithm $A$ with the following guarantee. Given as input an $M$-clause width-$w$ CNF $F$ 
over $\{0,1\}^n$ and parameters $w',\delta_\sand,\delta_\PRG,\delta_\mcount, \eta > 0$ and $p$ satisfying assumption (\ref{eq:p}) of Theorem~\ref{thm:GMR-SL}, 
\begin{enumerate}
\item $A$ runs in time 
\[ \exp(r_\SL(n,w,w',\eta,\delta_\sand) + r_\PRG(n,w',\delta_\PRG)) \cdot T_\mcount(n,M,w,\delta_\mcount),\]  where $r_\SL$ and $r_\PRG$ are as defined in Lemma~\ref{lem:key-lemma}, and $T_\mcount$ is as defined in Theorem~\ref{thm:GMR-count}. 
\item $A$ outputs a restriction $\tilde{\pi} \in \{0,1,\ast\}^n$ such that 
\begin{enumerate}
\item $|\tilde{\pi}^{-1}(\{0,1\})| \ge pn/2$, 
\item  $\tilde{\pi}$ approximately preserves the bias of $F$ from below: 
\[ \Ex_{\bx\leftarrow\calU}[(F\uhr \tilde{\pi})(\bx)] \ge \Ex_{\bx\leftarrow\calU}[F(\bx)] - (\delta_\PRG + \delta_\sand + \delta_\SL) - 2\delta_\mcount, \]
where 
\[ \delta_\SL =  \frac{2(\delta_\sand + \eta^{w'/4})}{p}.\] 
\end{enumerate} 
\end{enumerate}
\end{corollary}

\begin{proof}
The algorithm $A$ cycles through all (at most) $2^{r_\SL + r_\PRG}$ many restrictions $\pi$ in the support of the distribution $\calR_\gentle$ defined in Lemma~\ref{lem:key-lemma}, and for each one uses~\cite{GMR13}'s approximate counting algorithm in Theorem~\ref{thm:GMR-count} to approximate the bias of $F \uhr \pi$ to accuracy $\delta_\mcount$.  $A$ outputs the restriction $\tilde{\pi}$ for which  its estimate of the bias of $F \uhr \tilde{\pi}$ is the largest. 

The bound on the running time of $A$ is an immediate consequence of Lemma~\ref{lem:key-lemma} and Theorem~\ref{thm:GMR-count}, as is item {\it{$2(a)$}} in the claim. It remains to verify that $\tilde{\pi}$ satisfies {\it $2(b)$}. By Lemma~\ref{lem:key-lemma}, there is at least one $\pi^* \in \supp(\calR_\gentle)$ satisfying 
\[  \Ex_{\bx\leftarrow\calU}[(F\uhr \pi^*)(\bx)] \ge \Ex_{\bx\leftarrow\calU}[F(\bx)] - (\delta_\PRG + \delta_\sand + \delta_\SL). \] 
By the correctness of~\cite{GMR13}'s approximate counting algorithm, $A$'s estimate of the bias $F\uhr \pi^*$ is at least 
\[ \Ex_{\bx\leftarrow\calU}[(F\uhr \pi^*)(\bx)] - \delta_\mcount \ge  \Ex_{\bx\leftarrow\calU}[F(\bx)] - (\delta_\PRG + \delta_\sand + \delta_\SL) - \delta_\mcount, \] 
and hence so is its estimate of the bias of $F \uhr \tilde{\pi}$. Finally, again by the correctness of~\cite{GMR13}'s approximate counting algorithm, we conclude that the true bias of $F \uhr \tilde{\pi}$ is within $\delta_\mcount$ of $A$'s estimate, and hence 
\begin{align*}
\Ex_{\bx\leftarrow\calU}[(F\uhr \pi^*)(\bx)] &\ge  \bigg(\Ex_{\bx\leftarrow\calU}[F(\bx)] - (\delta_\PRG + \delta_\sand + \delta_\SL) - \delta_\mcount\bigg) - \delta_\mcount \\ 
&= \Ex_{\bx\leftarrow\calU}[F(\bx)] - (\delta_\PRG + \delta_\sand + \delta_\SL) - 2\delta_\mcount. 
\end{align*}
This completes the proof. 
\end{proof}

\subsection{Applying Corollary~\ref{cor:single-step}: setting of parameters} 
\label{sec:glorious}

We first introduce two more parameters $T \in \N$ and $\tau \in (0,1)$ to denote 
\[ T := \frac{2\ln n}{p} \quad \text{and} \quad \tau := \frac{\eps}{2T}. \]
Looking ahead, the semantics of $T$ and $\tau$ are as follows: each stage of our recursive search algorithm---a call to the subroutine in Corollary~\ref{cor:single-step}---fixes at least a $p/2$ fraction of the remaining coordinates (recall item $2(a)$ of Corollary~\ref{cor:single-step}), so $T$ is chosen so that after $T$ {stages} the number of unfixed coordinates is at most 
\[ n\cdot (1-p/2)^T = n\cdot (1-p/2)^{(2\ln n)/p} < 1, \]  
i.e.~we will have arrived at an actual assignment to the CNF $F$. Since $T$ is an upper bound on the number of calls to the subroutine in Corollary~\ref{cor:single-step}, we will set parameters so that the bias of $F$ is preserved to within an additive $\tau = \eps/2T$ in each call.  This ensures that the bias of $F$ remains at least 
\[ \eps - \tau\cdot T = \frac{\eps}{2}  > 0\] 
throughout, and hence the final assignment we arrive at is in fact a satisfying assignment of $F$.

With these definitions of $T$ and $\tau$ in hand, we will invoke the algorithm in Corollary~\ref{cor:single-step} with the following choice of parameters: 
\begin{align*}
 p  &= \left(\frac1{w\log((\log M)/\eps)}\right)^{2C\log w}, \\
 \eta &= \frac1{w\log((\log M)/\eps)}, \\
  w' &= 16C\log w + 4\log\left(\frac{192\ln M}{\eps}\right), \end{align*}
where $C > 0$ is the universal constant from Theorem~\ref{thm:GMR-SL}, and 
\[ \delta_\mcount = \frac{\tau}{3},\qquad \delta_\PRG = \frac{\tau}{6}, \qquad \delta_\sand = \frac{p\tau}{48}. \]

The next proposition justifies our choice of parameters: 

\begin{proposition}
\label{prop:verify}
For this choice of parameters, we have that 
\begin{enumerate}
\item $p, \eta$, and $\delta_\sand$ satisfy assumption (\ref{eq:p}) of Theorem~\ref{thm:GMR-SL}: 
\[ p \le \frac{\eta}{(w \log(1/\delta_{\sand}))^{C\log w}}. \] 
\item For $\delta_\SL = 2(\delta_\sand + \eta^{w'/4})/p$, 
\[ \delta_\PRG + \delta_\sand + \delta_\SL + 2\delta_\mcount \le \tau.\] 
\end{enumerate} 
\end{proposition} 

\begin{proof} 
For the first claim, we note that 
\begin{align*}
\log\left(\frac1{\delta_\sand}\right) &= \log\left(\frac{192 \ln n}{\eps p^2}\right)  \\
&= \log((\log n)/\eps) + 2\log(1/p) + O(1) \\
&= O(\log^2 w) (\log((\log M)/\eps)), 
\end{align*} 
and so indeed for $w$ larger than a suitable absolute constant, we have
\begin{align*}
\frac{\eta}{(w \log(1/\delta_\sand))^{C\log w}} &> \frac{\eta}{(w\log((\log M)/\eps))^{1.01C\log w}} \\
&= \left(\frac1{w\log((\log M)/\eps)}\right)^{1.01C\log w + 1} \\
&> \left(\frac1{w \log((\log M)/\eps)}\right)^{2C\log w} \ = \ p.
\end{align*} 
As for the second claim, by our choice of $\delta_\PRG = \tau/6$ and $\delta_\mcount = \tau/3$ the claimed bound is equivalent to 
\[ \delta_\sand + \delta_\SL \le \frac{\tau}{6}. \] 
Since $\delta_\sand < \delta_\SL$, it suffices to ensure that 
\[ \delta_\SL \le \frac{\tau}{12}, \qquad \text{or equivalently,} \qquad \delta_\sand + \eta^{w'/4} \le \frac{p\tau}{24}.\]
Recalling our choice of $\delta_\sand = p\tau / 48$, it remains to check that  
\[ \eta^{w'/4} \le \frac{p\tau}{48} = \frac{\eps p^2}{192\ln n}.\]
Indeed, 
\begin{align*}
\eta^{w'/4} &= \left(\frac1{w\log((\log M)/\eps)}\right)^{4C\log w + \log((192\ln M)/\eps)} \\
&<   \left(\frac1{w\log((\log M)/\eps)}\right)^{4C\log w} \cdot 2^{-\log((192\ln M)/\eps)}\\
&=  \frac{\eps p^2}{192\ln M} \leq \frac{\eps p^2}{192\ln n}  \tag*{(using $M \geq n$).}
\end{align*}
This completes the proof of the second claim.
\end{proof}

We note the following estimates for our choice of parameters when $w = O(\log(M/\eps))$ (recall Theorem~\ref{thm:main2} and in particular that $M \geq n$):
\begin{eqnarray}
\frac1{p} &=& (\log(M/\eps))^{O(\log \log(M/\eps))}  \label{eq:pval} \\
\log(1/\eta) &=& O(\log \log(M/\eps))   \label{eq:eta}\\
w' &=& O(\log((\log M)/\eps))  \label{eq:wprime}\\
\log(1/\delta) &=& O(\log((\log M)/\eps)) + O(\log\log(M/\eps))^2 . \quad \quad \text{(for $\delta \in \{\delta_\mcount, \delta_\PRG,\delta_\sand\}$)} \label{eq:deltas}
\end{eqnarray} 

Proposition~\ref{prop:verify} yields the following special case of Corollary~\ref{cor:single-step}:

\begin{corollary}[Corollary~\ref{cor:single-step} for our choice of parameters]
\label{cor:alg}
There is a deterministic algorithm $A$ with the following guarantee. Given as input an $M$-clause width-$w$ CNF $F$ over $\{0,1\}^n$,
\begin{enumerate}
\item $A$ runs in time 
\begin{align*} &\exp(r_\SL(n,w,w',\eta,\delta_\sand) + r_\PRG(n,w',\delta_\PRG)) \cdot T_\mcount(n,M,w,\delta_\mcount) \\
&=  \left({\frac {M} \eps}\right)^{\tilde{O}(\log((\log M)/\eps))^2}
\end{align*}
\item $A$ outputs a restriction $\tilde{\pi} \in \{0,1,\ast\}^n$ such that 
\begin{enumerate}
\item $|\tilde{\pi}^{-1}(\{0,1\})| \ge pn/2$, 
\item  $\ds\Ex_{\bx\leftarrow\calU}[(F\uhr \tilde{\pi})(\bx)] \ge \Ex_{\bx\leftarrow\calU}[F(\bx)] - \tau.$
\end{enumerate} 
\end{enumerate}
\end{corollary}

\section{Putting the pieces together: the overall search algorithm}
\label{sec:done} 

Using the results of the previous subsections we now prove Theorem \ref{thm:main2}.  The claimed algorithm is given as input a pair $(F,\eps)$; recall that from the theorem statement and as shown in Section \ref{sec:prelim}, we may assume that the CNF $F$ has $M \geq n$ clauses each of width at most $w=O(\log(M/\eps))$.  

The algorithm proceeds for at most $T = (2\ln n)/p$ iterative stages (where $p$ is as defined in Section
\ref{sec:glorious}) as follows.  In the $t$-th stage it operates on a CNF formula 
$F \uhr (\tilde{\pi}^0 \circ \cdots \circ \tilde{\pi}^{t-1})$; the first stage is the $(t=1)$-th stage and we take $\tilde{\pi}^0$ to be the trivial restriction which assigns $\ast$ to each of the $n$ input variables, so $F \uhr \tilde{\pi}^0$ is simply the input CNF $F$.  Before starting the first stage, the algorithm records the values of parameters $w,w',\eta,\delta_\sand,$ and 
$\delta_\PRG$. (Observe that all of these values $w,w',\eta,\delta_\sand,\delta_\PRG$ are defined solely in terms of $M$ and $\eps$, \ignore{ a dependence only on $M$ and $\eps$ and with no dependence on ``$n$'' (thanks to our assumption that $M=n^{\Omega(1)}$), }see Equations 
(\ref{eq:eta}), (\ref{eq:wprime}) and (\ref{eq:deltas}); these values will never change during the execution of the algorithm.)

Stage~1 is carried out as follows: 

\begin{itemize}

\item Let $n_1$ denote
the number of variables that are alive under restriction $\tilde{\pi}^0$, which in stage 1 is $n_1 = n$.  \ignore{(Later, in the general $t$-th stage, these values will be defined using $n_t$, the number of variables that are alive under the restriction $\tilde{\pi}^0 \circ \cdots \circ \tilde{\pi}^{t-1}$.)} The algorithm compute the seed lengths $r_{\SL,1} := r_\SL(n_1,w,w',\eta,\delta_\sand)$ and $r_{\PRG,1} := r_\PRG(n_1,w',\delta_\PRG)$.  

\item Then the algorithm executes the deterministic algorithm $A$ from Corollary \ref{cor:alg} on the $n_1$-variable function 
$F \uhr \tilde{\pi}^0$.  The algorithm produces a restriction $\tilde{\pi}^1 \in \{0,1,\ast\}^{n_1}$ with the properties described in 2(a) and 2(b) of Corollary \ref{cor:alg}.

\end{itemize}

The general $t$-th stage of the algorithm is carried out in a similar way:

\begin{itemize}

\item Let $n_t$ denote the number of variables that are alive under the restriction $\tilde{\pi}^0 \circ \cdots \circ \tilde{\pi}^{t-1} \in \{0,1,\ast\}^n$.  \ignore{Using this value of $n_t$ in place of ``$n$'', it computes new values for 
$w',\eta,\delta_\sand, \delta_\PRG$ using the expressions given in Section \ref{sec:glorious}.  (Observe that the value of
$w$, which does not have any dependence on ``$n$,'' is unchanged.)} The algorithm computes the seed lengths $r_{\SL,t} := r_\SL(n_t,w,w',\eta,\delta_\sand)$ and $r_{\PRG,t} := r_\PRG(n_t,w',\delta_\PRG)$ which are appropriate for the pseudorandom switching lemma and pseudorandom generators for $n_t$-variable functions.

\item Then the algorithm executes the deterministic algorithm $A$ from Corollary \ref{cor:alg} \emph{on the $n_t$-variable CNF 
$F \uhr (\tilde{\pi}^0 \circ \cdots \circ \tilde{\pi}^{t-1})$}.  The algorithm produces a restriction $\tilde{\pi}^t \in \{0,1,\ast\}^{n_t}$ with the properties described in 2(a) and 2(b) of Corollary \ref{cor:alg}.

\end{itemize}

We may view the restriction $\tilde{\pi}^0 \circ \cdots \circ \tilde{\pi}^{t}$ as belonging to $\{0,1,\ast\}^n$.  If $\tilde{\pi}^0 \circ \cdots \circ \tilde{\pi}^{t}$ belongs to $\{0,1\}^n$ (leaves no variables free) then the algorithm halts
and outputs $\tilde{\pi}^0 \circ \cdots \circ \tilde{\pi}^{t}$, otherwise it increments $t$ and proceeds to the next stage.

\medskip

It remains to establish correctness; this is easy given Corollary \ref{cor:alg}.
A crucial aspect of the algorithm is that in the $t$-th stage it works on the $n_t$-variable CNF $F \uhr (\tilde{\pi}^0 \circ \cdots \circ \tilde{\pi}^{t-1})$.  Thanks to part 2(a) of Corollary \ref{cor:alg}, this implies that each value of $n_t$ is at most $n(1-p/2)^{t-1}$, so consequently after at most $T$ stages the algorithm will indeed obtain a restriction $\tilde{\pi}^0 \circ \cdots \circ \tilde{\pi}^{t} \in \{0,1\}^n$ and halt as desired.  For the running time of the algorithm, it follows from part (1) of Corollary \ref{cor:alg} that the running time of each of the (at most) $T = (2\ln n)/p$ stages is upper bounded by $\left(M/\eps\right)^{\tilde{O}(\log((\log M)/\eps))^2}$ and hence this is also an upper bound on the running time of the entire algorithm (recalling the bound on $p$ from (\ref{eq:pval})).  Finally, from the discussion at the start of Section \ref{sec:glorious}, we have that the bias of $F(\tilde{\pi}^0 \circ \cdots \circ \tilde{\pi}^{t})$ is greater than zero, and hence $\tilde{\pi}^0 \circ \cdots \circ \tilde{\pi}^{t}$ is a satisfying assignment as desired.  This concludes the proof of Theorem~\ref{thm:main2}.

\bibliography{allrefs}{}
\bibliographystyle{alpha}

\end{document}